\documentclass[proceedings]{stacs}
\stacsheading{2009}{51--62}{Freiburg}
\firstpageno{51}

\usepackage[latin1]{inputenc}
\usepackage[T1]{fontenc}
\usepackage{pstricks,pst-node,pstricks-add}
\usepackage{enumerate}

\usepackage{stmaryrd,mathrsfs}
\usepackage{amsmath,amssymb,latexsym}
\usepackage{times}

\theoremstyle{plain}
\theoremstyle{definition}

\newcommand{\ov}[1]{\overline{#1}}
\newcommand{\vek}[1]{\vec{#1}}

\renewcommand*{\geq}{\ensuremath{\geqslant}}
\renewcommand*{\leq}{\ensuremath{\leqslant}}

\newcommand{\twodots}{.\,.\,}

\newcommand{\raus}[1]{}

\newcommand{\set}[1]{\ensuremath{\{ #1 \}}}
\newcommand{\setc}[2]{\set{ #1 : #2}}

\newcommand*{\A}{\ensuremath{\mathcal{A}}}

\newenvironment{mi}{\begin{itemize}}{\end{itemize}}

\newenvironment{mea}{\begin{enumerate}[\mbox{\quad }(a)]}{\end{enumerate}}

\newcommand{\myparagraph}[1]{\subsubsection*{{\bf #1}}}

\newcommand{\DI}{\ensuremath{\mathbb{D}}}
\newcommand{\Disj}{\ensuremath{\textit{Disj}}}

\newcommand{\aut}{\text{mp2s-automaton}}
\newcommand{\auta}{\text{mp2s-automata}}

\newcommand{\kf}{\ensuremath{k_{\scriptscriptstyle f}}}
\newcommand{\kb}{\ensuremath{k_{\scriptscriptstyle b}}}

\newcommand{\adv}{\ensuremath{\textsl{advance}}}
\newcommand{\stay}{\ensuremath{\textsl{stay}}}

\newcommand{\eos}{\ensuremath{\textsl{end}}}

\newcommand{\vv}{\ensuremath{v}}

\newcommand{\config}{\ensuremath{\textit{config}}}

\begin{document}
\title[Lower Bounds for Multi-Pass Processing of Multiple Data Streams]
{Lower Bounds for Multi-Pass Processing of \\ Multiple Data Streams}

\author{Nicole Schweikardt}{Nicole Schweikardt}
\address{Institut f\"ur Informatik, 
  Goethe-Universit\"at Frankfurt am Main, \newline
  Robert-Mayer-Str.\ 11--15, D-60325 Frankfurt am Main, 
  Germany}  
\email{schweika@informatik.uni-frankfurt.de}  
\urladdr{http://www.informatik.uni-frankfurt.de/\~{}schweika}  


\vspace{-1ex}



\begin{abstract}
  \noindent 
  This paper gives a brief overview of computation models for data stream processing, and
  it introduces a new model for multi-pass processing of
  multiple streams, the so-called \emph{\auta{}}. Two algorithms for solving the 
  set disjointness problem with these automata are presented.
  The main technical contribution of this paper is the proof of a 
  lower bound on the size of memory and the number of heads that are
  required for solving the set disjointness problem with \auta{}.
\end{abstract}

\maketitle


\section{Introduction}\label{section:introduction}

In the basic data stream model, the input consists of a stream of data items 
which can be read only sequentially, one after the other.
For processing these data items, a memory buffer of limited size is available.
When designing data stream algorithms, one aims at 
algorithms whose memory size is far smaller than the size of the input.

Typical application areas for which
data stream processing is relevant are, e.g., 
IP network traffic analysis, mining text message streams, or processing
meteorological data generated by sensor networks.
Data stream algorithms are also used to support query optimization
in relational database systems.
In fact, virtually all query optimization methods in 
relational database systems rely on information about
the number of distinct values of an attribute or the self-join size of a relation 
--- and these pieces of information have to be maintained while the database is updated. 
Data stream algorithms for accomplishing this task have been introduced in the seminal
paper \cite{alomatsze99}.

Most parts of the data stream literature deal with the task of performing
\textbf{one pass over a single stream}. For a detailed overview
on algorithmic techniques for this scenario we refer to \cite{Muthukrishnan-DataStream-Algos}.
 \emph{Lower bounds} on the size of  memory needed for solving a problem by a 
one-pass algorithm are usually obtained by applying methods from 
\emph{communication complexity} 
(see, e.g., \cite{alomatsze99,henragraj99}).
In fact, for many concrete problems it is known that the memory needed 
for solving the problem by a deterministic one-pass algorithm is at 
least linear in the size $n$ of the input.
For some of these problems, however,
\emph{randomized} one-pass algorithms can still compute good \emph{approximate} answers 
while using memory of size sublinear in $n$.
Typically, such algorithms are based on \emph{sampling}, 
i.e., only a ``representative'' portion of the data is taken into account, 
and \emph{random projections}, i.e., only a rough ``sketch'' of the data is stored in memory.
See \cite{Muthukrishnan-DataStream-Algos,ChristianSohler-Survey} for a comprehensive survey 
of according algorithmic techniques and for pointers to the literature.

Also the generalization where \textbf{multiple passes over a single stream} are performed, 
has received considerable attention in the literature. 
Techniques for proving lower bounds in this scenario can be found, 
e.g., in 
\cite{henragraj99,GuhaMcGregor-ICALP08,DBLP:journals/dcg/ChanC07,DBLP:conf/focs/GalG07,munpat80}.

A few articles also deal with the task of \textbf{processing several streams in parallel}.
For example, the authors of \cite{Shalem-BarYossef-icde2008} consider 
algorithms which perform one pass over several streams.
They introduce a new model of multi-party communication complexity that is suitable for
proving lower bounds 
on the amount of memory necessary for one-pass algorithms on multiple
streams.
In \cite{Shalem-BarYossef-icde2008}, these results are used for determining 
the exact space complexity of processing particular XML twig queries.
\\
In recent years, the database community has also addressed the issue of 
designing general-purpose \emph{data stream management systems}
and query languages that are suitable for new application areas where multiple data streams
have to be processed in parallel. To get an overview of this research area, 
\cite{bbdmw02} is a good starting point. Foundations for a theory of \emph{stream queries}
have been laid in \cite{DBLP:conf/dbpl/GurevichLB07}.
Stream-based approaches have also
been examined in detail in connection with \emph{XML query processing and validation},
see, e.g.\ the papers \cite{SegoufinVianu_PODS02,SegoufinSirangelo_ICDT07,SuciuEtAl_TODS04,ChanFelberEtAl-VLDBJ2002,barfonjos04,barfonjos05,grokocschwe05a}.

The \emph{finite cursor machines} (FCMs, for short) of 
\cite{FCM-paper} are a computation model for performing
\textbf{multiple passes over multiple streams}.
FCMs were introduced as an abstract model of database query processing. 
Formally, they are defined in the framework of 
\emph{abstract state machines}.
Informally, they can be described as follows:
The input for an FCM is a relational database, each relation of which is 
represented by a \emph{table}, i.e., an ordered list of rows, where each 
row corresponds to a tuple in the relation.
Data elements are viewed as
``indivisible'' objects that can be manipulated by a number of ``built-in'' operations. 
This feature is very convenient to model standard operations on data types 
like integers, floating point numbers,
or strings, which may all be part of the universe of data elements.
FCMs can operate in a finite number of \emph{modes} using an 
\emph{internal memory} in which they can
store bitstrings. They access each relation through a finite number of 
\emph{cursors}, each of which can read one row of a table at any time.
The model incorporates certain \emph{streaming} or \emph{sequential processing} 
aspects by imposing a restriction on the movement of the cursors: 
They can move on the tables only sequentially in one direction. 
Thus, once the last cursor has left a row of a table, 
this row can never be accessed again during the computation.
Note, however, that several cursors can be moved asynchronously over the same table 
at the same time, and thus, entries in different, possibly far apart, 
regions of the table can be read and processed simultaneously.

A common feature of the computation models mentioned so far in this paper is that the
input streams are \emph{read-only} streams that cannot be modified during a
pass. Recently, also
\textbf{stream-based models for external memory processing}
have been proposed, among them 
the \emph{StrSort model} \cite{aggdatrajruh04,Ruhl-PhD}, 
the \emph{W-Stream} model \cite{WStream-Paper}, and the
model of \emph{read/write streams} 
\cite{groschwe05a,grokocschwe05a,groherschwe06,BJR-STOC07,BH-FOCS08}.
In these models, several passes may be performed over a single stream or over
several streams in parallel, and during a pass, the content of the stream may be
modified.
\\
\par

A detailed introduction to \emph{algorithms on data streams}, respectively, to the 
related area of \emph{sublinear algorithms} can be found in
\cite{Muthukrishnan-DataStream-Algos,ChristianSohler-Survey}.
A survey of \emph{stream-based models for external memory processing} and of 
methods for proving 
\emph{lower bounds} in these models is given in \cite{Schweikardt-PODS07survey}.
A database systems oriented overview of so-called \emph{data stream systems} can
be found in \cite{bbdmw02}.
For a list of \emph{open problems} in the area of data streams we refer to
\cite{Kanpur-2006-ListOfOpenProblems}.
\\
\par

In the remainder of this article, a new computation model for multi-pass processing
of multiple streams is introduced: the \emph{\auta{}}. 
In this model,
(read-only) streams can be processed by forward scans as well as 
backward scans, and several ``heads'' can be used to perform several passes 
over the streams in parallel. 
After fixing the basic notation in Section~\ref{section:preliminaries}, the computation
model of \auta{} is introduced in Section~\ref{section:mpms-automata}.
In Section~\ref{section:main-result}, we consider the \emph{set disjointness problem} 
and prove upper bounds as well as lower bounds
on the size of memory and the number of heads that are necessary for solving
this problem with an \aut{}.
Section~\ref{section:conclusion} concludes the paper by pointing out some directions
for future research.

\section{Basic notation}\label{section:preliminaries}

If $f$ is a function from the set of non-negative integers to the set
of reals, we shortly write $f(n)$ instead of $\lceil f(n)\rceil$ (where
$\lceil x \rceil$ denotes the smallest integer $\geq x$).
We write $\lg n$ to denote the logarithm of $n$ with respect to base 2.
For a set $\DI$ we write $\DI^*$ to denote the set of all 
finite strings over alphabet $\DI$. We view $\DI^*$  as 
the set of all finite \emph{data streams} that can be built from 
elements in $\DI$.
For a stream $\vek{S}\in\DI^*$ write $|\vek{S}|$ to denote the
length of $\vek{S}$, and we write $s_i$ to denote the element 
in $\DI$ that occurs at the $i$-th position in $\vek{S}$, i.e., 
$\vek{S}= s_1 s_2\cdots s_{|\vek{S}|}$.

\section{A computation model for multi-pass processing of multiple streams}
\label{section:mpms-automata}

In this section, we fix a computation model for multi-pass processing
of multiple streams. 
The model is quite powerful: 
Streams can be processed by forward scans as well as backward scans, and
several ``heads'' can be used to perform several passes over the stream 
in parallel.
For simplicity, we restrict attention to the case
where just \emph{two} streams are processed in parallel.
Note, however, that it is straightforward to generalize the model to an
arbitrary number of streams.

The computation model, called
\emph{\auta}\footnote{\emph{``mp2s''} stands for 
\underline{m}ulti-\underline{p}ass processing of \underline{2} 
\underline{s}treams}, can be described as follows:
Let $\DI$ be a set, and let $m,\kf,\kb$ be integers with 
$m\geq 1$ and $\kf,\kb\geq 0$.
An \vspace{1ex} 
\begin{center}
  \emph{\aut{} $\A$ with parameters $(\DI,m,\kf,\kb)$}  \vspace{1ex}
\end{center}
receives as input two streams $\vek{S}\in\DI^*$ and $\vek{T}\in\DI^*$.
The automaton's memory consists of $m$ different states
(note that this corresponds to a memory buffer consisting of 
$\lg m$ bits). The automaton's state space is denoted by $Q$.
We assume that $Q$ contains a designated \emph{start state} and
that there is a designated subset $F$ of $Q$ of 
so-called \emph{accepting states}.

On each of the input streams $\vek{S}$ and $\vek{T}$, the automaton
has $\kf$ heads that process the stream from left to right 
(so-called \emph{forward heads}) and $\kb$ heads that process the stream 
from right to left (so-called \emph{backward heads}).
The heads are allowed to move asynchronously.
We use $k$ to denote the total number of heads, i.e., $k= 2\kf + 2\kb$.

\begin{figure}[h!tbp]
 \begin{center}
  \fbox{$
   \begin{array}{cll}
     \DI & : & 
      \text{set of \emph{data items} of which input 
         streams $\vek{S}$ and $\vek{T}$
         are composed}
    \\[1ex]
      m  & : & 
      \text{size of the automaton's \emph{state space} $Q$
         \ (this corresponds to $\lg m$ bits of memory)}
    \\[1ex]
      \kf & : & 
      \text{number of \emph{forward heads} available on each input stream}
    \\[1ex]
      \kb & : & 
       \text{number of \emph{backward heads} available on each input stream}
    \\[1ex]
      k   & : & 2\kf + 2\kb \quad \text{(total number of heads)}
   \end{array}
  $}
   \caption{The meaning of the parameters $(\DI,m,\kf,\kb)$ of an \aut.}
 \end{center}
\end{figure}

In the \emph{initial configuration} of $\A$ on input $(\vek{S},\vek{T})$,
the automaton is in the \emph{start state},
all \emph{forward} heads on $\vek{S}$ and $\vek{T}$ are placed on the leftmost
element in the stream, i.e., $s_1$ resp.\ $t_1$, and all
\emph{backward} heads are placed on the rightmost element in the stream, 
i.e., $s_{|\vek{S}|}$ resp.\ $t_{|\vek{T}|}$. 

During each computation step, depending on
(a) the current state (i.e., the current content of the automaton's
memory) and (b) the elements of $\vek{S}$ and $\vek{T}$ at the 
current head positions,
a deterministic transition function determines
(1) the next state (i.e., the new content of the automaton's memory) 
and (2) which of the $k$ heads should be advanced to the 
next position (where forward heads are advanced one step to the right, and
backward heads are advanced one step to the left). Formally, the 
transition function can be specified in a straightforward way by a function 
\[
   \delta\ : \ Q \times (\DI\cup \set{\eos})^k \ \longrightarrow \ Q \times \set{\adv, \stay}^k
\] 
where $Q$ denotes the automaton's state space, and $\eos$ is a special
symbol (not belonging to $\DI$) which indicates that a head has 
reached the end of the stream (for a forward head this means that the
head has been advanced beyond the rightmost element of the stream, and for
a backward head this means that the head has been advanced beyond the
leftmost element of the stream).

The automaton's computation on input $(\vek{S},\vek{T})$ ends as soon as
each head has passed the entire stream. The input is \emph{accepted} if the
automaton's state then belongs to the set $F$ of accepting states, and
it is \emph{rejected} otherwise.
\\
\par

The computation model of \auta{} is closely related to the \emph{finite 
cursor machines} of \cite{FCM-paper}. In both models, several streams can
be processed in parallel, and several heads (or, ``cursors'')
may be used to perform several ``asynchronous'' passes over the same stream in 
parallel. In contrast to the \auta{} of the present paper, finite cursor machines
were introduced as an abstract model for database query processing, and their
formal definition in \cite{FCM-paper} is presented in the framework of 
\emph{abstract state machines}.

Note that \auta{} can be viewed as a generalization of other models for 
one-pass or multi-pass processing of streams.
For example, the scenario of \cite{Shalem-BarYossef-icde2008}, where a single
pass over two streams is performed, is captured by an \aut{} where
1~forward head and no backward heads are available on each stream. 
Also, the scenario where $p$ consecutive passes
of each input stream are available (cf., e.g., \cite{henragraj99}), can 
be implemented by an
\aut{}: just use $p$ forward heads and $0$ backward heads, and let
the $i$-th head wait at the first position of the stream until the
$(i{-}1)$-th head has reached the end of the stream.

\section{The set disjointness problem}\label{section:main-result}

Throughout Section~\ref{section:main-result} we consider a particular
version of the \emph{set disjointness problem} where, 
for each integer $n\geq 1$, 
\ \(
   \DI_n \ := \ \set{\,a_1,b_1,\ \ldots, \, a_n,b_n \,}
\) \
is a fixed set of $2n$ data items.
We write $\Disj_n$ to denote the following decision problem:
The input consists of two streams $\vek{S}$ and $\vek{T}$ over $\DI_n$
with $|\vek{S}| = |\vek{T}| = n$.
The goal is to decide whether the sets $\set{s_1,\ldots, s_n}$ and 
$\set{t_1,\ldots,t_{n}}$ are disjoint.

An \aut{} \emph{solves} the problem $\Disj_n$ if, for all valid
inputs to $\Disj_n$ (i.e., all $\vek{S},\vek{T}\in \DI^*$ with $|\vek{S}| = |\vek{T}| = n$),
it accepts the input if, and only if, the corresponding
sets are disjoint.

\subsection{Two upper bounds for the set disjointness problem}\label{subsec:upperbounds}

It is straightforward to see that the problem $\Disj_n$ can be solved
by an \aut{} with $2^{2n}$~states and a single forward head on each of the
two input streams: During a first phase, the
head on $\vek{S}$ processes $\vek{S}$ and stores, in the automaton's
current state, the subset of $\DI_n$ that has been seen while processing $\vek{S}$.
Afterwards, the head on $\vek{T}$ processes $\vek{T}$ and checks whether
the element currently seen by this head belongs to the subset of $\DI_n$ that
is stored in the automaton's state. 
Clearly, $2^{2n}$ states suffice for this task, since $|\DI_n|=2n$. We thus
obtain the following trivial upper bound:

\begin{proposition}\label{prop:upper-bound-trivial}
 $\Disj_n$ can be solved by an \aut{} with parameters
 $(\DI_n,2^{2n},1,0)$.
\end{proposition}

The following result shows that, at the expense of increasing the number 
of forward heads on each stream to $\sqrt{n}$, the memory consumption can
be reduced exponentially:

\begin{proposition}\label{prop:upper-bound-sqrt}
 $\Disj_n$ can be solved by an \aut{} with parameters
 $(\DI_n,n{+}2,\sqrt{n},0)$.\footnote{To be precise, the proof shows that already 
    $n+2-\sqrt{n}$ states suffice.}
\end{proposition}
\begin{proof}
The automaton proceeds in two phases. 

The goal in \emph{Phase~1} is to move, 
for each $i\in\set{1,\ldots,\sqrt{n}\,}$, the $i$-th head on $\vek{S}$ onto
the $\big((i{-}1)\sqrt{n}\,\, +1\big)$-th position in $\vek{S}$. 
This way, after having finished \emph{Phase~1}, the heads 
partition $\vek{S}$ into $\sqrt{n}$ sub-streams, each of
which has length $\sqrt{n}$.
Note that $n+1-\sqrt{n}$ states suffice for accomplishing this: The automaton simply
stores, in its state, the current position of the rightmost head(s) on $\vek{S}$.
It starts by
leaving head~1 at position $1$ and moving the remaining heads on $\vek{S}$
to the right until position $\sqrt{n}+1$ is reached. 
Then, it leaves head~2 at position $\sqrt{n}+1$ and proceeds by moving the
remaining heads to the right until position $2\sqrt{n}+1$ is reached, etc.

During \emph{Phase~2}, the automaton checks whether the two sets are disjoint.
This is done in $\sqrt{n}$ sub-phases. During the $j$-th sub-phase, the
$j$-th head on $\vek{T}$ processes $\vek{T}$ from left to right and compares
each element in $\vek{T}$ with the elements on the current positions of
the $\sqrt{n}$ heads on $\vek{S}$. When the $j$-th head on $\vek{T}$ has 
reached the end of the stream, each of the heads on $\vek{S}$ is moved one
step to the right. This finishes the $j$-th sub-phase.
Note that \emph{Phase~2} can be accomplished by using just 2 states:
By looking at the combination of heads on $\vek{T}$ that have already 
passed the entire stream, the automaton can tell which sub-phase it is 
currently performing. Thus, for \emph{Phase~2} we just need
one state for indicating that the automaton is in \emph{Phase~2}, 
and an additional state for storing that the
automaton has discovered already that the two sets are \emph{not} disjoint.
\end{proof}

\subsection{Two lower bounds for the set disjointness problem}\label{subsec:lowerbound}

We first show a lower bound for \auta{} where only forward heads are available:

\begin{theorem}\label{thm:lower-bound-forward}
 For all integers $n$, $m$, $\kf$, such that,
 for \ $k = 2\kf$ \ and  \ $\vv = \kf^2 + 1$,
 \[
    k^2\cdot \vv \cdot \lg(n{+}1) \ + \ k\cdot \vv \cdot \lg m \ + \ v \cdot (1+\lg v)
    \ \ \leq \ \ 
    n\,, 
 \]
 the problem $\Disj_n$ cannot be solved 
 by any \aut{} with parameters $(\DI_n,m,\kf,0)$.
\end{theorem}
\begin{proof}
Let $n$, $m$, and $\kf$ 
be chosen such that they meet the theorem's assumption.
For contradiction, let us assume that $\A$ is an \aut{} with parameters
$(\DI_n,m,\kf,0)$ that solves the problem $\Disj_n$.

Recall that
\ \(
   \DI_n \ = \ \set{\,a_1,b_1,\ \ldots, \, a_n,b_n \,}
\) \
is a fixed set of $2n$ data items.
Throughout the proof we will restrict attention to input streams $\vek{S}$ and
$\vek{T}$ which are enumerations of the elements in a set \vspace{1ex}
\[
   A^I \quad := \quad \setc{a_i}{i\in I} \ \cup \ \setc{b_i}{i\in \ov{I}}  \vspace{1ex}
\]
for arbitrary $I\subseteq \set{1,\twodots,n}$ and its complement
$\ov{I}:=\set{1,\twodots,n}\setminus I$.
\\
Note that for all $I_1,I_2\subseteq \set{1,\twodots,n}$ we have \vspace{1ex}
\begin{equation}\label{equ:1}
  A^{I_1} \text{ and } A^{I_2} \text{ are disjoint} \ \iff \ I_2=\ov{I_1}. \vspace{1ex}
\end{equation}
For each $I\subseteq \set{1,\twodots,n}$ we let $\vek{S}^I$ be the stream of length $n$
which is defined as follows: For each $i\in I$, it carries  
data item $a_i$ at position $i$; and for each $i\not\in I$, it carries data item $b_i$ at
position $i$.
The stream $\vek{T}^I$ contains the same data items as $\vek{S}^I$, but in the
opposite order: For each $i\in I$, it
carries data item $a_i$ at position $n-i+1$; and for each $i\not\in I$, it carries
data item $b_i$ at position $n-i+1$.

For sets $I_1,I_2\subseteq \set{1,\twodots,n}$, we write $D(I_1,I_2)$ to denote the input instance
$\vek{S}^{I_1}$ and $\vek{T}^{I_2}$ for the problem $\Disj_n$.
From (\ref{equ:1}) and our assumption that the \aut{} $\A$ solves $\Disj_n$, we obtain 
that 
\begin{equation}\label{equ:2}
  \A \text{ \ accepts \ } D(I_1,I_2) \ \iff \ I_2 = \ov{I_1}. \vspace{1ex}
\end{equation}
Throughout the remainder of this proof, our goal is to find two sets
$I,I'\subseteq \set{1,\twodots,n}$ such that \vspace{1ex}
\begin{enumerate}
 \item
   $I\neq I'$, \ and \vspace{1ex}
 \item
   the accepting run of $\A$ on $D(I,\ov{I})$ is ``similar'' to the
   accepting run of $\A$ on $D(I', \ov{I'})$, so that the two runs can be combined into
   an accepting run of $\A$ on $D(I,\ov{I'})$ (later on in the proof, we will see
   what ``similar'' precisely means).
  \vspace{1ex}
\end{enumerate}
Then, however,
the fact that $\A$ accepts input $D(I,\ov{I'})$ contradicts (\ref{equ:2}) and thus
finishes the proof of Theorem~\ref{thm:lower-bound-forward}.

For accomplishing this goal, we let 
\begin{equation}\label{eq:def-v}
  \vv \ \ := \ \ \kf^2 +1
\end{equation}
be 1 plus the number of pairs of heads on the two streams. We subdivide the set
$\set{1,\twodots,n}$ into $\vv$ consecutive blocks $B_1,\ldots, B_{\vv}$ of equal size 
$\frac{n}{\vv}$. I.e., for each $j\in\set{1,\twodots,\vv}$, block $B_j$ consists of the
indices in $\set{\, (j{-}1) \frac{n}{\vv} + 1, \,\ldots, \, j \frac{n}{\vv}\,}$.

We say that a pair $(h_S,h_T)$ of heads of $\A$ \emph{checks block $B_j$} during the run
on input $D(I_1,I_2)$
if, and only if, at some point in time during the run, there exist
$i,i'\in B_j$ such that head $h_S$ is on element
$a_i$ or $b_i$ in $\vek{S}^{I_1}$ and head $h_T$ is on element
$a_{i'}$ or $b_{i'}$ in $\vek{T}^{I_2}$.

Note that each pair of heads can check at most one block, since only forward heads
are available and the data items in $\vek{T}^{I_2}$ are arranged in the reverse
order (with respect to the indices $i$ of elements $a_i$ and $b_i$) than in $\vek{S}^{I_1}$.
Since there are $\vv$ blocks, but only $\vv-1$ pairs $(h_S,h_T)$ of heads on the
two streams, we know that for each $I_1,I_2\subseteq \set{1,\twodots,n}$ there exists
a block $B_j$ that is \emph{not checked} during $\A$'s run on $D(I_1,I_2)$.

In the following, we determine a set 
$X\subseteq \setc{I}{I\subseteq\set{1,\twodots,n}}$ with $|X|\geq 2$ 
such that for all $I,I'\in X$, item (2) of our goal is satisfied.
We start by using a simple averaging argument to find a
$j_0\in\set{1,\twodots,\vv}$ and 
a set $X_0\subseteq \setc{I}{I\subseteq\set{1,\twodots,n}}$
such that \vspace{1ex}
\begin{mi}
 \item 
   for each $I\in X_0$, block $B_{j_0}$ is not checked during $\A$'s run on input
   $D(I,\ov{I})$, \ and  \vspace{1ex}
 \item 
   $|X_0| \geq \frac{2^n}{\vv}$. \vspace{1ex}
\end{mi}
For the remainder of the proof we fix $\hat{B}:= B_{j_0}$.
\\
We next choose a sufficiently large set $X_1\subseteq X_0$ in which everything outside
block $\hat{B}$ is fixed:\\
A simple averaging argument shows that there is a
$X_1\subseteq X_0$ and a $\hat{I}\subseteq \set{1,\twodots,n}\setminus \hat{B}$ such that
\vspace{1ex}
\begin{mi}
 \item
   for each $I\in X_1$, \ \ $I\setminus \hat{B} = \hat{I}$, \ \ and \vspace{1ex}
 \item 
   $|X_1| \ \geq \ \frac{|X_0|}{2^{n-\frac{n}{\vv}}} \ \geq \ 2^{\frac{n}{\vv}-\lg \vv}$.
   \vspace{1ex}
\end{mi}
We next identify a set $X_2\subseteq X_1$ such that for all $I,I'\in X_2$ the runs of
$\A$ on $D(I,\ov{I})$ and $D(I',\ov{I'})$ are ``similar'' in a sense suitable for
item~(2) of our goal. To this end, for each head $h$ of $\A$ we let
$\config_h^I$ be the \emph{configuration} (i.e., the current state and the absolute
positions of all the heads) in the run of $\A$ on input $D(I,\ov{I})$ at the 
particular point in time where head $h$ has just left block $\hat{B}$ (i.e., head
$h$ has just left the last element $a_i$ or $b_i$ with $i\in \hat{B}$ that it can
access).
We let $\config^I$ be the ordered tuple of the configurations $\config_h^I$ for all
heads $h$ of $\A$.
Note that the number of possible configurations $\config_h^I$ is 
\ $\leq \ m\cdot (n{+}1)^{k}$, since $\A$ has $m$ states and since each of the
$k = 2\kf$ heads can be at one out of $n{+}1$ possible positions in its input stream.
Consequently, the number of possible $k$-tuples $\config^I$ of configurations is 
\ $\leq \ \big( m\cdot (n{+}1)^{k} \big)^{k}$.
\\
A simple averaging argument thus yields a tuple $c$ of configurations and
a set $X_2\subseteq X_1$ such that \vspace{1ex}
\begin{mi}
 \item 
   for all $I\in X_2$, \ $\config^I = c$, \ and \vspace{1ex}
 \item
   $|X_2| \ \geq \ \frac{|X_1|}{( m\cdot (n{+}1)^{k})^{k}} \ \geq \  
    2^{\, \frac{n}{\vv} - \lg \vv \ - \ k \lg m \ - \ k^2 \lg(n+1)}.
   $ \vspace{1ex}
\end{mi}
Using the theorem's assumption on the 
numbers $n$, $m$, and $\kf$, one obtains that $|X_2|\geq 2$.
Therefore, we can find two sets $I,I'\in X_2$ with $I\neq I'$.

To finish the proof of Theorem~\ref{thm:lower-bound-forward}, it remains to show that
the runs of $\A$ on $D(I,\ov{I})$ and on $D(I',\ov{I'})$ 
can be combined into a run of $\A$ on $D(I,\ov{I'})$ such that $\A$ (falsely) accepts
input $D(I,\ov{I'})$. To this end let us summarize what we know about $I$ and $I'$ in 
$X_2$: \vspace{1ex}

\begin{mea}
 \item $I$ and $I'$ only differ in block $\hat{B}$.
  \vspace{1ex}
 \item Block $\hat{B}$ is not checked during $\A$'s runs on $D(I,\ov{I})$ and 
   on $D(I',\ov{I'})$. I.e., while any head on $\vek{S}^I$ (resp.\ $\vek{S}^{I'}$) is
   at an element $a_i$ or $b_i$ with $i\in \hat{B}$, no head on $\vek{T}^{\ov{I}}$
   (resp.\ $\vek{T}^{\ov{I'}}$) is on an element $a_{i'}$ or $b_{i'}$ with $i'\in \hat{B}$.
  \vspace{1ex}
 \item Considering $\A$'s runs on $D(I,\ov{I})$ and on $D(I',\ov{I'})$, each
  time a head leaves the last position in $\hat{B}$ that it can access, both runs are 
  are in exactly the same configuration. I.e., they are in the same state, and all
  heads are at the same absolute positions in their input streams.
  \vspace{1ex}
\end{mea}

\noindent
Due to item~(a),
$\A$'s run on input $D(I,\ov{I'})$ starts in the same way as the runs on
$D(I,\ov{I})$ and $D(I',\ov{I'})$: As long as no head has reached an element in
block $\hat{B}$, the automaton has not yet seen any difference between $D(I,\ov{I'})$ on
the one hand and $D(I,\ov{I})$ and $D(I',\ov{I'})$ on the other hand.

At some point in time, however, some head $h$ will enter block $\hat{B}$, 
i.e., it will enter the first element $a_i$ or $b_i$ with $i\in \hat{B}$ that 
it can access.
The situation then is as follows: \vspace{1ex}
\begin{mi}
 \item 
   If $h$ is a head on $\vek{S}^{I}$, then, due to item~(b), no head on 
   $\vek{T}^{\ov{I'}}$ is at an element in $\hat{B}$. 
   Therefore, until head $h$ leaves block $\hat{B}$, $\A$ will go through the 
   same sequence of configurations as in its run on input $D(I,\ov{I})$. 
   Item~(c) ensures that when $h$ leaves block $\hat{B}$, $\A$ is in the same configuration
   as in its runs on $D(I,\ov{I})$ and on $D(I',\ov{I'})$.
  \vspace{1ex}
 \item
   Similarly, if $h$ is a head on $\vek{T}^{\ov{I'}}$, then, due to item~(b), no head on 
   $\vek{S}^{I}$ is at an element in $\hat{B}$. 
   Therefore, until head $h$ leaves block $\hat{B}$, $\A$ will go through the 
   same sequence of configurations as in its run on input $D(I',\ov{I'})$. 
   Item~(c) ensures that when $h$ leaves block $\hat{B}$, $\A$ is in the same configuration
   as in its runs on $D(I',\ov{I'})$ and on $D(I,\ov{I})$.
  \vspace{1ex}
\end{mi}

\noindent
In summary, in $\A$'s run on $D(I,\ov{I'})$, each time a head $h$ has just left the
last element in block $\hat{B}$ that it can access, it is in exactly the same 
configuration as in $\A$'s runs on $D(I,\ov{I})$ and on $D(I',\ov{I'})$ at the points
in time where head $h$ has just left the
last element in block $\hat{B}$ that it can access.
After the last head has left block $\hat{B}$, $\A$'s run on $D(I,\ov{I'})$ finishes
in exactly the same way as $\A$'s runs on $D(I,\ov{I})$ and $D(I',\ov{I'})$. 
In particular, it accepts $D(I,\ov{I'})$ (since it accepts $D(I,\ov{I})$ and
$D(I',\ov{I'})$). This, however, is a contradiction to (\ref{equ:2}).
Thus, the proof of Theorem~\ref{thm:lower-bound-forward} is complete.
\end{proof}

\medskip

\begin{remark}
Let us compare the lower bound from Theorem~\ref{thm:lower-bound-forward} 
with the upper bound of Proposition~\ref{prop:upper-bound-sqrt}: The upper bound
tells us that $\Disj_n$ can be solved by an \aut{} with $n{+}2$ states and
$\sqrt{n}$ forward heads on each input stream. The lower bound
implies (for large enough $n$) that
if just $\sqrt[5]{n}$ forward heads are available on each stream, not even
$2^{\sqrt[3]{n}}$ states suffice for solving the problem $\Disj_n$ with an
\aut{}.
\end{remark}

\begin{remark}
A straightforward calculation shows that the assumptions of 
Theorem~\ref{thm:lower-bound-forward} are satisfied, for example,
for all sufficiently large 
integers $n$ and all integers $m$ and $\kf$ with
\ $4\kf  \leq  \sqrt[4]{\frac{n}{\lg n}}$ \ and 
\ $\lg m \leq  \frac{n}{4\kf \cdot (\kf^2+1)}$.
\end{remark}

\par

Theorem~\ref{thm:lower-bound-forward} can be generalized 
to the following lower bound for \auta{} where also backward heads are available:

\begin{theorem}\label{thm:lower-bound-general}
 For all $n$, $m$, $\kf$, $\kb$ such that,
 for \ $k = 2\kf+ 2\kb$ \ and  \ $\vv = (\kf^2 + \kb^2 + 1)\cdot (2\kf\kb + 1)$,%
 \[
    k^2\cdot \vv \cdot \lg(n{+}1) \ + \ k\cdot \vv \cdot \lg m \ + \ v \cdot (1+\lg v)
    \ \ \leq \ \ 
    n\,, 
 \]
 the problem $\Disj_n$ cannot be solved 
 by any \aut{} with parameters $(\DI_n,m,\kf,\kb)$.
\end{theorem}
\begin{proof}
The overall structure of the proof is the same as in the proof of
Theorem~\ref{thm:lower-bound-forward}.
We consider the same sets $A^I$, for all $I\subseteq\set{1,\twodots,n}$.
The stream $\vek{S}^I$ is chosen in the same way as in the proof of
Theorem~\ref{thm:lower-bound-forward}, i.e., 
for each $i\in I$, the stream $\vek{S}^I$ carries  
data item $a_i$ at position $i$; and for each $i\not\in I$, it carries data item $b_i$ at
position $i$.

Similarly as in the proof of Theorem~\ref{thm:lower-bound-forward}, the
stream $\vek{T}^I$ contains the same data items as $\vek{S}^I$. Now, however, the
order in which the elements occur in $\vek{T}^I$ is a bit more elaborate.
For fixing this order, we choose the following parameters: \vspace{1ex}
\begin{equation}
  \vv_1 \ := \ \kf^2 + \kb^2 + 1\ ,
  \qquad
  \vv_2 \ := \ 2\kf \kb + 1\ ,
  \qquad
  \vv \ := \ v_1\cdot v_2\ .
  \vspace{1ex}
\end{equation}
We subdivide the set
$\set{1,\twodots,n}$ into $\vv_1$ consecutive blocks $B_1,\ldots, B_{\vv}$ of equal size 
$\frac{n}{\vv_1}$. I.e., for each $j\in\set{1,\twodots,\vv_1}$, block $B_j$ consists of the
indices in $\set{\, (j{-}1) \frac{n}{\vv_1} + 1, \,\ldots, \, j \frac{n}{\vv_1}\,}$.
\\
Afterwards, we further subdivide each block $B_j$ into $v_2$ consecutive subblocks of 
equal size $\frac{n}{\vv}$. These subblocks are denoted $B_{j}^1,\ldots,B_{j}^{\vv_2}$.
Thus, each subblock $B_{j}^{j'}$ consists of the indices in  
$\set{\, (j{-}1) \frac{n}{\vv_1} + (j'{-}1) \frac{n}{\vv} + 1, \,\ldots, \, 
 (j{-}1) \frac{n}{\vv_1} +  j' \frac{n}{\vv}\,}$.

Now let $\pi$ be the permutation of $\set{1,\twodots,n}$ which maps, for all $j,r$ with
$1\leq j\leq \vv_1$ and $1\leq r\leq \frac{n}{\vv_1}$, \ element 
$(j{-}1) \frac{n}{\vv_1} + s$ \ onto \ element \ 
$(\vv_1 {-} j) \frac{n}{\vv_1} + s$.
Thus, $\pi$ maps elements in block $B_j$ onto elements in block $B_{\vv_1-j+1}$, and
inside these two blocks, $\pi$ maps the elements of subblock $B_{j}^{j'}$ onto
elements in subblock $B_{\vv_1-1+1}^{j'}$. 
Note that $\pi$ reverses the blocks $B_j$ in order, but it does \emph{not} reverse the 
order of the subblocks $B_j^{j'}$.

Finally, we are ready to fix the order in which the elements in $A^I$ occur in the
stream $\vek{T}^I$: For each $i\in I$, the stream $\vek{T}^I$ carries
data item $a_i$ at position $\pi(i)$; and for each $i\not\in I$, it carries
data item $b_i$ at position $\pi(i)$.

In the same way as in the proof of Theorem~\ref{thm:lower-bound-forward}, we
write $D(I_1,I_2)$ to denote the input instance $\vek{S}^{I_1}$ and $\vek{T}^{I_2}$.

A pair of heads $(h_S,h_T)$ is called \emph{mixed} if one of the heads is a
forward head and the other is a backward head.
Since $\pi$ reverses the order of the blocks $B_1,\twodots,B_{\vv_1}$, 
it is straightforward to see that every \emph{non-mixed} pair of heads can
check at most one of the blocks $B_1,\twodots,B_{\vv_1}$. Since there are
$\vv_1$ blocks, but only $(\vv_1-1)$ non-mixed pairs of heads, we know that for
all $I_1,I_2\subseteq \set{1,\twodots,n}$ there exists a block $B_j$ that is 
\emph{not checked} by any non-mixed pair of heads during $\A$'s run on input
$D(I_1,I_2)$.

The same averaging argument as in the proof of Theorem~\ref{thm:lower-bound-forward}
thus tells us that there is a
$j_1\in\set{1,\twodots,\vv_1}$ and 
a set $X'_{0}\subseteq \setc{I}{I\subseteq\set{1,\twodots,n}}$
such that \vspace{1ex}
\begin{mi}
 \item 
   for each $I\in X'_{0}$, block $B_{j_1}$ is not checked by any non-mixed pair of
   heads during $\A$'s run on input
   $D(I,\ov{I})$, \ and  \vspace{1ex}
 \item 
   $|X'_{0}| \geq \frac{2^n}{\vv_1}$. \vspace{1ex}
\end{mi}
From our particular choice of $\pi$, it is straightforward to see that every
\emph{mixed} pair of heads can check at most one of the subblocks
$B_{j_1}^1,\ldots,B_{j_1}^{\vv_2}$. Since there are $\vv_2$ such subblocks, but
only $(\vv_2 -1)$ mixed pairs of heads, there must be a $j_2\in\set{1,\twodots,\vv_2}$ 
and a set $X_0\subseteq X'_0$ such that \vspace{1ex}
\begin{mi}
 \item 
   for each $I\in X_{0}$, subblock $B_{j_1}^{j_2}$ is not checked by any pair of
   heads during $\A$'s run on input
   $D(I,\ov{I})$, \ and  \vspace{1ex}
 \item 
   $|X_{0}| \ \geq \ \frac{|X'_0|}{\vv_2} \ \geq \ \frac{2^n}{\vv}$. \vspace{1ex}
\end{mi}
For the remainder of the proof we fix $\hat{B}:= B_{j_1}^{j_2}$, and we let
$k:= 2\kf + 2\kb$ denote the total number of heads. Using these notations, 
the rest of the proof can be taken vertatim from the proof of 
Theorem~\ref{thm:lower-bound-forward}.
\end{proof}

The proof of
Theorem~\ref{thm:lower-bound-general} is implicit in \cite{FCM-paper}
(see Theorem 5.11 in \cite{FCM-paper}). There, however, the proof is formulated
in the terminology of a different
machine model, the so-called \emph{finite cursor machines}.

\section{Final remarks}\label{section:conclusion}

Several questions concerning the computational power of 
\auta{} occur naturally. On a technical level, it would be nice
to determine the exact complexity of the set disjointness problem with
respect to \auta{}.
In particular: Is the upper bound provided by Proposition~\ref{prop:upper-bound-sqrt}
optimal?
Can backward scans significantly help for solving the set disjointness problem?
Are $\sqrt{n}$ heads really necessary for
solving the set disjointness problem when only a sub-exponential number of states
are available?

A more important task, however, is to consider also randomized versions
of \auta{}, to design efficient randomized approximation algorithms for particular
problems, and to develop techniques for proving lower bounds in the randomized model.

\myparagraph{Acknowledgement.}
I would like to thank Georg Schnitger for
helpful comments on an earlier version of this paper.

\end{document}